\documentclass[a4paper,envcountsame]{llncs}

\usepackage[english]{babel}
\usepackage[utf8]{inputenc}

\usepackage{hyperref}
\usepackage{enumerate}
\usepackage{amsmath}
\usepackage{amssymb}

\usepackage{subcaption}

\usepackage{listings}
\lstdefinelanguage{pseudo}{
	morekeywords={if, then, else, while, foreach, do, assume, assert, let, return},
	sensitive=false,
	morecomment=[l]{//},
	morecomment=[s]{/*}{*/},
	morestring=[b]",
}
\AtBeginDocument{

}

\usepackage{etoolbox}
\AtEndEnvironment{example}{%
	\qed
}

\newcommand{\loopt}{{\scriptstyle\mathrm{LOOP}}} 
\newcommand{\stemt}{{\scriptstyle\mathrm{STEM}}} 
\def\diag{\mathrm{diag}} 
\newcommand{\abovebelow}[2]{
\left(\begin{smallmatrix}
{#1} \\ {#2}
\end{smallmatrix}\right)
} 

\title{Geometric Nontermination Arguments}
\titlerunning{Geometric Nontermination Arguments}
\author{Jan Leike\inst{1} \and Matthias Heizmann\inst{2}}
\authorrunning{J.\ Leike \and M.\ Heizmann}
\institute{
  Australian National University \and
  University of Freiburg
}

\begin{document}

\maketitle

\begin{abstract}
We present a new kind of nontermination argument,
called \emph{geometric nontermination argument}.
The geometric nontermination argument is a finite representation
of an infinite execution that has the form of
a sum of several geometric series.
For so-called linear lasso programs we can decide the existence of
a geometric nontermination argument using a nonlinear algebraic $\exists$-constraint.
We show that a deterministic conjunctive loop program
with nonnegative eigenvalues is nonterminating
if an only if there exists a geometric nontermination argument.
Furthermore,
we present an evaluation that demonstrates that our method is feasible in practice.
\end{abstract}


\section{Introduction}
\label{sec:introduction}

The problem whether a program is terminating is undecidable in general.
One way to approach this problem in practice is to analyze the existence of
termination arguments and nontermination arguments.
The existence of a certain termination argument like,
e.g, a linear ranking function, is decidable~\cite{PR04,BAG13} and implies termination.
However, if we cannot find a linear ranking function we cannot conclude nontermination.
Vice versa, the existence of a certain nontermination argument like,
e.g, a linear recurrence set~\cite{GHMRX08}, is decidable and implies nontermination
however, if we cannot find such a recurrence set we cannot conclude termination.

In this paper we present a new kind of termination argument which we call
\emph{geometric nontermination argument (GNTA)}.
Unlike a recurrence set, a geometric nontermination argument
does not only imply nontermination,
it also explicitly represents an infinite program execution.
An infinite program execution that is represented by a geometric nontermination argument
can be written as a pointwise sum of several geometric series.
We show that such an infinite execution exists for each
deterministic conjunctive loop program that is nonterminating
and whose transition matrix has only nonnegative eigenvalues.

We restrict ourselves to linear lasso programs.
A lasso program consists of a single while loop that is preceded by straight-line code.
The name refers to the lasso shaped form of the control flow graph.
Usually, linear lasso programs do not occur as stand-alone programs.
Instead, they are used as a finite representation of
an infinite path in a control flow graph.
For example, in (potentially spurious) counterexamples in termination analysis~\cite{CPR06,BCF13,HLNR10,KSTTW08,KSTW10,PR04TI,PodelskiR05,HHP14},
stability analysis~\cite{CFKP11,PW07},
cost analysis~\cite{AAGP11,GZ10},
or the verification of temporal properties~\cite{CookKV11,CookKP15,DietschHLP15} for programs.

We present a constraint based approach that allow us to check whether a
linear conjunctive lasso program has a geometric nontermination argument
and to synthesize one if it exists.

Our analysis is motived by the
probably simplest form of an infinite executions,
namely infinite execution where the same state is always repeated.
We call such a state a fixed point.
For lasso programs we can reduce the check for the existence of a fixed point
to a constraint solving problem as follows.
Let us assume that the stem and the loop of the lasso program are given as
a formulas over primed and unprimed variables $\stemt(\vec x, \vec x')$
and $\loopt(\vec x, \vec x')$.
The infinite sequence
$\vec s_0, \bar{\vec s}, \bar{\vec s} , \bar{\vec s},\ldots$
is an nonterminating execution of the lasso program iff the assignment
$\vec x_0\mapsto \vec s_0, \bar{\vec x}\mapsto \bar{\vec s}$
is a satisfying assignment for the constraint
$\stemt(\vec x_0, \bar{\vec x})\land \loopt(\bar{\vec x}, \bar{\vec x})$.
In this paper, we present a constraint that is not only satisfiable if the program has
a fixed point, it is also satisfiable if the program has a nonterminating
execution that can be written as a pointwise sum of geometric series.

\begin{figure}[t]
\begin{center}
\begin{subfigure}[b]{0.3\textwidth}
\begin{lstlisting}
b := 1;
while (a+b >= 3):
  a := 3*a + 1;
  b := nondet();
\end{lstlisting}
\caption{ }
\label{fig:IntroExA}
\end{subfigure}
\begin{subfigure}[b]{0.3\textwidth}
\begin{minipage}{35mm}
\begin{lstlisting}
b := 1;
while (a+b >= 3):
  a := 3*a - 2;
  b := 2*b;
\end{lstlisting}
\end{minipage}\hspace{5mm}
\caption{ }
\label{fig:IntroExB}
\end{subfigure}
\begin{subfigure}[b]{0.3\textwidth}
\begin{minipage}{35mm}
\begin{lstlisting}
b := 1;
while (a+b >= 4):
  a := 3*a + b;
  b := 2*b;
\end{lstlisting}
\end{minipage}
\caption{ }
\label{fig:IntroExC}
\end{subfigure}
\end{center}
\caption{
Three nonterminating linear lasso programs.
Each has an infinite execution which is either a geometric series or
a pointwise sum of geometric series.
The first lasso program is nondeterministic because the variable \texttt{b}
gets some nondeterministic value in each iteration.
}
\label{fig:introexample}
\end{figure}

Let us motivate the representation of infinite executions
as sums of geometric series in three steps.
The program depicted in \autoref{fig:IntroExA}
shows a lasso program which does not have a fixed point but the following
infinite execution.

\[
\abovebelow{2}{0},
\abovebelow{2}{1},
\abovebelow{7}{1},
\abovebelow{22}{1},
\abovebelow{67}{1},
\dots
\]
We can write this infinite execution as a a geometric series
where for $k>1$ the $k$-th state is the sum
$\vec{x_1} + \sum_{i=0}^k \lambda^i \vec{y}$,
where we have
$\vec{x_1} = \abovebelow{2}{1}$,
$\vec y = \abovebelow{5}{0}$, and
$\lambda = 2$.
The state $\vec{x_1}$ is the state before the loop was executed before
the first time and intuitively $\vec y$ is the direction in which the
execution is moving initially and $\lambda$ is the speed at which the
execution continues to move in this direction.

Next, let us consider the lasso program depicted in \autoref{fig:IntroExB}
which has the following infinite execution.
\[
\abovebelow{2}{0},
\abovebelow{2}{1},
\abovebelow{4}{4},
\abovebelow{10}{8},
\abovebelow{28}{16},
\dots
\]
We cannot write this execution as a geometric series as we did above.
Intuitively, the reason is that the values of both variables are
increasing at different speeds and hence this execution is not
moving in a single direction.
However, we can write this infinite execution as a sum of geometric series
where for $k>1$ the $k$-th state can be written as a sum
$\vec{x_1} + \sum_{i=0}^{k} Y
\big(
\begin{smallmatrix}\lambda_1&0\\0&\lambda_2\end{smallmatrix}
\big)^i \vec1$,
where we have
$\vec{x_1} = \abovebelow{2}{1}$,
$\vec Y = \left(
\begin{matrix}
2 & 0 \\
0 & 1
\end{matrix}
\right)$,
$\lambda_1=3,\lambda_2=2$
and $\vec1$ denotes the column vector of ones.
Intuitively, our execution is moving in two different directions at different speeds.
The directions are reflected by the column vectors of $Y$,
the values of $\lambda_1$ and $\lambda_2$ reflect the respective speeds.

Let us next consider the lasso program in \autoref{fig:IntroExC}
which has the following infinite execution.
\[
\abovebelow{3}{0},
\abovebelow{3}{1},
\abovebelow{10}{2},
\abovebelow{32}{4},
\abovebelow{100}{8},
\dots
\]
We cannot write this execution as a pointwise sum of geometric series in the
form that we used above.
Intuitively, the problem is that one of the initial directions contributes at
two different speeds to the overall progress of the execution.
However, we can write this infinite execution as a pointwise sum of geometric series
where for $k>1$ the $k$-th state can be written as a sum
$\vec{x_1} + \sum_{i=0}^{k} Y
\big(
\begin{smallmatrix}\lambda_1&\mu\\0&\lambda_2\end{smallmatrix}
\big)^i \vec1$,
where we have
$\vec{x_1} = \abovebelow{3}{1}$,
$\vec Y = \left(
\begin{matrix}
4 & 3 \\
0 & 1
\end{matrix}
\right)$,
$\lambda_1=3,\lambda_2=2,\mu=1$
and $\vec1$ denotes the column vector of ones.
We call the tuple $(\vec{x_0}, \vec{x_1}, Y, \lambda_1, \lambda_2, \mu)$
which we use as a finite representation for the infinite execution a
\emph{geometric nontermination argument}.

In this paper, we formally introduce the notion of
a geometric nontermination argument for linear lasso programs
(\autoref{sec:gnta}) and we prove that each
nonterminating deterministic conjunctive linear loop program
whose transition matrix has only nonnegative real eigenvalues
has a  geometric nontermination argument, i.e.,
each such nonterminating linear loop program has an infinite execution
which can be written as a sum of geometric series (\autoref{sec:completeness}).


\section{Preliminaries}
\label{sec:preliminaries}

We denote vectors $\vec x$ with bold symbols and
matrices with uppercase Latin letters.
Vectors are always understood to be column vectors,
$\vec1$ denotes a vector of ones,
$\vec0$ denotes a vector of zeros (of the appropriate dimension), and
$\vec{e_i}$ denotes the $i$-th unit vector.
A list of notation can be found on page~\pageref*{app:notation}.

\subsection{Linear Lasso Programs}
\label{ssec:linear-lassos}

In this work, we consider linear lasso programs,
programs that consist of a program step and a single loop.
We use binary relations over the program's states to
define the stem and the loop transition relation.
Variables are assumed to be real-valued.

We denote by $\vec x$
the vector of $n$ variables $(x_1, \ldots, x_n)^T \in \mathbb{R}^n$
corresponding to program states, and
by $\vec{x'} = (x_1', \ldots, x_n')^T \in \mathbb{R}^n$
the variables of the next state.

\begin{definition}[Linear Lasso Program]\label{def:linear-lassos}
A (conjunctive) \emph{linear lasso program} $L = (\stemt, \loopt)$
consists of two binary relations defined by
formulas with the free variables $\vec x$ and $\vec{x'}$
of the form
\[
A \abovebelow{\vec x}{\vec{x'}} \leq \vec{b}
\]
for some matrix $A \in \mathbb{R}^{n \times m}$
and some vector $\vec{b} \in \mathbb{R}^m$.
\end{definition}

A \emph{linear loop program} is
a linear lasso program $L$ without stem, i.e.,
a linear lasso program such that the relation $\stemt$ is equivalent to $true$.

\begin{definition}[Deterministic Linear Lasso Program]\label{def:det-linear-lassos}
A linear loop program $L$ is called \emph{deterministic} iff its loop transition $\loopt$ can be written in the following
form
\[
(\vec{x}, \vec{x}') \in \loopt
\;\Longleftrightarrow\;
G \vec{x} \leq \vec{g} \;\land\; \vec{x}' = M \vec{x} + \vec{m}
\]
for some matrices $G \in \mathbb{R}^{n \times m}$, $M \in \mathbb{R}^{n \times n}$, and
vectors $\vec{g} \in \mathbb{R}^m$ and $\vec{m} \in \mathbb{R}^n$.
\end{definition}

\begin{definition}[Nontermination]\label{def:nontermination}
A linear lasso program $L$ is \emph{nonterminating} iff
there is an infinite sequence of states $\vec{x_0}, \vec{x_1}, \ldots$,
called an \emph{infinite execution of $L$}, such that
$(\vec{x_0}, \vec{x_1}) \in \stemt$ and
$(\vec{x_t}, \vec{x_{t+1}}) \in \loopt$ for all $t \geq 1$.
\end{definition}

\subsection{Jordan Normal Form}
\label{ssec:jordan-nf}

Let $M \in \mathbb{R}^{n \times n}$ be a real square matrix.
If there is an invertible square matrix $S$ and a diagonal matrix $D$ such that
$M = SDS^{-1}$, then $M$ is called \emph{diagonalizable}.
The column vectors of $S$ form the basis over which $M$ has diagonal form.
In general, real matrices are not diagonalizable.
However, every real square matrix $M$ with real eigenvalues
has a representation which is almost diagonal,
called \emph{Jordan normal form}.
This is a matrix that is zero except for the eigenvalues on the diagonal and
one superdiagonal containing ones and zeros.

Formally, a Jordan normal form is a matrix
$J = \mathrm{diag}(J_{i_1}(\lambda_1), \ldots, J_{i_k}(\lambda_k))$
where $\lambda_1, \ldots, \lambda_k$ are the eigenvalues of $M$ and
the real square matrices $J_i(\lambda) \in \mathbb{R}^{i \times i}$
are \emph{Jordan blocks},
\[
J_i(\lambda) :=
\left(
\begin{matrix}
\lambda & 1 & 0 & \ldots & 0 & 0 \\
0 & \lambda & 1 & \ldots & 0 & 0\\
\vdots &  & & \ddots & & \vdots \\
0 & 0 & 0 & \ldots & \lambda & 1 \\
0 & 0 & 0 & \ldots & 0 & \lambda
\end{matrix}
\right).
\]
The subspace corresponding to each distinct eigenvalue is called
\emph{generalized eigenspace} and
their basis vectors \emph{generalized eigenvectors}.

\begin{theorem}[Jordan Normal Form]\label{thm:jordan-nf}
For each real square matrix $M \in \mathbb{R}^{n \times n}$ with real eigenvalues,
there is an invertible real square matrix $V \in \mathbb{R}^{n \times n}$
and a Jordan normal form $J \in \mathbb{R}^{n \times n}$
such that
$M = V J V^{-1}$.
\end{theorem}


\section{Geometric Nontermination Arguments}
\label{sec:gnta}

Fix a conjunctive linear lasso program $L = (\stemt, \loopt)$ and
let $A \in \mathbb{R}^{n \times m}$ and $\vec{b} \in \mathbb{R}^m$ define
the loop transition such that
\[
(\vec{x}, \vec{x}') \in \loopt
\;\Longleftrightarrow\;
A \abovebelow{\vec{x}}{\vec{x}'} \leq \vec{b}.
\]

\begin{definition}[Geometric Nontermination Argument]\label{def:gnta}
A tuple \\
$(\vec{x_0}, \vec{x_1}, \vec{y_1}, \ldots, \vec{y_k},
\lambda_1, \ldots, \lambda_k, \mu_1, \ldots, \mu_{k-1})$
is called a \emph{geometric nontermination argument}
for the linear lasso program $L = (\stemt, \loopt)$
iff all of the following statements hold.
\begin{enumerate}
\itemsep2mm
\setlength{\itemindent}{4em}
\item[(domain)]\label{itm:gnta-domain}
  $\vec{x_0}, \vec{x_1}, \vec{y_1}, \ldots, \vec{y_k} \in \mathbb{R}^n$, and
  $\lambda_1, \ldots, \lambda_k, \mu_1, \ldots, \mu_{k-1} \geq 0$
\item[(initiation)]\label{itm:gnta-init}
  $(\vec{x_0}, \vec{x_1}) \in \stemt$
\item[(point)]\label{itm:gnta-point}
  $A \abovebelow{\vec{x_1}}{\vec{x_1} + \sum_i \vec{y_i}} \leq \vec{b}$
\item[(ray)]\label{itm:gnta-ray}
  $A \abovebelow{\vec{y_1}}{\lambda_1 \vec{y_1}} \leq 0$ and
  $A \abovebelow{\vec{y_i}}{\lambda_i \vec{y_i} + \mu_{i-1} \vec{y_{i-1}}} \leq 0$
  for $1 < i \leq k$.
\end{enumerate}
The number $k \geq 0$ is
the \emph{size} of the geometric nontermination argument.
\end{definition}

The existence of a geometric nontermination argument can be checked using an SMT solver.
The constraints given by \hyperref[itm:gnta-domain]{(domain)},
\hyperref[itm:gnta-init]{(init)}, \hyperref[itm:gnta-point]{(point)},
\hyperref[itm:gnta-ray]{(ray)} are nonlinear algebraic constraints.
The satisfiability of these constraints is decidable.
Moreover, if the linear lasso program is given as a deterministic update,
we can compute its eigenvalues.
If the eigenvalues are known,
we can assign values to $\lambda_1, \ldots, \lambda_k$ and
the constraints become linear and can thus be decided efficiently.

\begin{proposition}[Soundness]\label{prop:soundness}
If there is a geometric nontermination argument for a linear lasso program $L$,
then $L$ is nonterminating.
\end{proposition}
\begin{proof}
Define $Y := (\vec{y_1} \ldots \vec{y_k})$ denote the matrix containing the vectors $\vec{y_i}$ as columns,
and define the matrix
\begin{equation}\label{eq:def-U}
U :=
\left(
\begin{matrix}
\lambda_1 & \mu_1 & 0 & \ldots & 0 & 0 \\
0 & \lambda_2 & \mu_2 & \ldots & 0 & 0\\
\vdots &  & & \ddots & & \vdots \\
0 & 0 & 0 & \ldots & \lambda_{n-1} & \mu_{n-1} \\
0 & 0 & 0 & \ldots & 0 & \lambda_n
\end{matrix}
\right).
\end{equation}
Following \autoref{def:nontermination}
we show that the linear lasso program $L$
has the infinite execution
\begin{equation}\label{eq:infinite-execution}
\vec{x_0},\;
\vec{x_1},\;
\vec{x_1} + Y \vec1,\;
\vec{x_1} + Y \vec1 + Y U \vec1,\;
\vec{x_1} + Y \vec1 + Y U \vec1 + Y U^2 \vec1, \ldots
\end{equation}
From \hyperref[itm:gnta-init]{(init)} we get $(\vec{x_0}, \vec{x_1}) \in \stemt$.
It remains to show that
\begin{equation}\label{eq:execution-induction1}
\left(
\vec{x_1} + \sum_{j=0}^{t-1} Y U^j \vec1,\;
\vec{x_1} + \sum_{j=0}^t Y U^j \vec1
\right) \in \loopt
\text{ for all $t \in \mathbb{N}$.}
\end{equation}
According to \hyperref[itm:gnta-domain]{(domain)}
the matrix $U$ has only nonnegative entries, so
the same holds for the matrix $Z := \sum_{j=0}^{t-1} U^j$.
Hence $Z \vec1$ has only nonnegative entries and thus
$YZ\vec1$ can be written as $\sum_{i=1}^k \alpha_i \vec{y_i}$
for some $\alpha_i \geq 0$.
We multiply the inequality number $i$ from \hyperref[itm:gnta-ray]{(ray)}
with $\alpha_i$ and get
\begin{equation}\label{eq:scaled-ray}
A \abovebelow{
	\alpha_i \vec{y_i}
}{
	\alpha_i \lambda_i \vec{y_i} + \alpha_i \mu_{i-1} \vec{y_{i-1}}
}
\leq 0.
\end{equation}
where we use the convenience notation $\vec{y_0} := 0$ and $\mu_0 := 0$.
Now we sum \eqref{eq:scaled-ray} for all $i$ and
add \hyperref[itm:gnta-point]{(point)} to get
\begin{equation}\label{eq:execution-induction2}
A \abovebelow{
	\vec{x_1} + \sum_i \alpha_i \vec{y_i}
}{
	\vec{x_1} + \sum_i \vec{y_i} + \sum_i (\alpha_i \lambda_i \vec{y_i}
	+ \alpha_i \mu_{i-1} \vec{y_{i-1}})
}
\leq \vec b.
\end{equation}
By definition of $\alpha_i$, we have
\[
    \vec{x_1} + \sum_{i=1}^k \alpha_i \vec{y_i}
~=~ \vec{x_1} + Y Z \vec1
~=~ \vec{x_1} + \sum_{j=0}^{t-1} Y U^j \vec1
\]
and
\begin{align*}
  \vec{x_1} + \sum_{i=1}^k \vec{y_i}
  + \sum_{i=1}^k (\alpha_i \lambda_i \vec{y_i}
  + \alpha_i \mu_{i-1} \vec{y_{i-1}})
&= \vec{x_1} + Y\vec1 + \sum_{i=1}^k \alpha_i YU e_i \\
&= \vec{x_1} + Y\vec1 + YU Z\vec1 \\
&= \vec{x_1} + \sum_{j=0}^t Y U^j \vec1.
\end{align*}
Therefore \eqref{eq:execution-induction1} and \eqref{eq:execution-induction2}
are the same, which concludes this proof.
\qed
\end{proof}

\begin{example}[Closed Form of the Infinite Execution]
The following is the closed form of
the state $\vec{x_t} = \sum_{k=0}^t Y U^k \vec 1$
in the infinite execution \eqref{eq:infinite-execution}.
Let $U =: N + D$ where $N$ is a nilpotent matrix
and $D$ is a diagonal matrix.
\begin{align*}
   Y U^k \vec 1
= Y \left( \sum_{i=0}^n \binom{k}{i} N^i D^{k-i} \right) \vec 1
= \sum_{j=1}^n y_j \sum_{i=0}^{n - j + 1} \binom{k}{i} \lambda_{n-j-i}^{k-i} \prod_{\ell=j}^{j+i-1} \mu_\ell
\tag*{$\Diamond$}
\end{align*}
\renewcommand{\qed}{}
\end{example}


\section{Completeness}
\label{sec:completeness}

First we show that a linear loop program has a GNTA
if it has is a bounded infinite execution.
In the next section we use this to prove our completeness result.

\subsection{Bounded Infinite Executions}
\label{ssec:bounded-infinite-executions}

Let $|\cdot|: \mathbb{R}^n\rightarrow \mathbb{R}$ denote some norm.
We call an infinite execution $(\vec{x}_t)_{t \geq 0}$ \emph{bounded} iff
there is a real number $d \in \mathbb{R}$ such that
the norm of each state is bounded by $d$,
i.e., $|\vec{x}_t|\leq d$ for all $t$
(in $\mathbb{R}^n$ the notion of boundedness is independent of the choice of the norm).

\begin{lemma}[Fixed Point]\label{lem:fixed-point}
Let $L = (true, \loopt)$ be a linear loop program.
The linear loop program $L$ has a bounded infinite execution
if and only if
there is a fixed point $\vec{x}^\ast \in \mathbb{R}^n$ such that
$(\vec{x}^\ast, \vec{x}^\ast) \in \loopt$.
\end{lemma}
\begin{proof}
If there is a fixed point $\vec{x}^\ast$, then the loop has the infinite bounded
execution $\vec{x}^\ast, \vec{x}^\ast, \ldots$.
Conversely, let $(\vec{x}_t)_{t \geq 0}$ be an infinite bounded execution.
Boundedness implies that there is an $d \in \mathbb{R}$ such that $|\vec{x}_t| \leq d$ for all $t$.
Consider the sequence $\vec{z}_k := \frac{1}{k} \sum_{t=1}^k \vec{x}_t$.
\begin{align*}
| \vec{z}_k - \vec{z}_{k+1} |
&= \left| \frac{1}{k} \sum_{t=1}^k \vec{x}_t - \frac{1}{k+1} \sum_{t=1}^{k+1} \vec{x}_t \right|
 = \frac{1}{k(k+1)} \left| (k+1) \sum_{t=1}^k \vec{x}_t -  k \sum_{t=1}^{k+1} \vec{x}_t \right| \\
&= \frac{1}{k(k+1)} \left| \sum_{t=1}^k \vec{x}_t - k \vec{x}_{k+1} \right|
\leq \frac{1}{k(k+1)} \left( \sum_{t=1}^k |\vec{x}_t| + k |\vec{x}_{k+1}| \right) \\
&\leq \frac{1}{k(k+1)} (k\cdot d + k\cdot d)
= \frac{2d}{k+1} \longrightarrow 0 \text{ as } k \to \infty.
 \end{align*}
Hence the sequence $(\vec{z}_k)_{k \geq 1}$ is a Cauchy sequence
and thus converges to some $\vec{z}^\ast \in \mathbb{R}^n$.
We will show that $\vec{z}^\ast$ is the desired fixed point.

For all $t$, the polyhedron $Q := \{ \abovebelow{\vec{x}}{\vec{x}'} \mid A \abovebelow{\vec{x}}{\vec{x}'} \leq b \}$ contains $\abovebelow{\vec{x}_t}{\vec{x}_{t+1}}$ and is convex.
Therefore for all $k \geq 1$,
$$
\frac{1}{k} \sum_{t=1}^k \abovebelow{\vec{x}_t}{\vec{x}_{t+1}} \in Q.
$$
Together with
$$
\abovebelow{\vec{z}_k}{\frac{k+1}{k} \vec{z}_{k+1}}
= \frac{1}{k} \abovebelow{\vec{0}}{\vec{x}_1} + \frac{1}{k} \sum_{t=1}^k \abovebelow{\vec{x}_t}{\vec{x}_{t+1}}
$$
we infer
$$
\left( \abovebelow{\vec{z}_k}{\frac{k+1}{k} \vec{z}_{k+1}} - \frac{1}{k} \abovebelow{\vec{0}}{\vec{x}_1} \right) \in Q,
$$
and since $Q$ is topologically closed we have
\[
\abovebelow{\vec{z}^\ast}{\vec{z}^\ast} =
\lim_{k \to \infty}
  \left( \abovebelow{\vec{z}_k}{\frac{k+1}{k} \vec{z}_{k+1}} - \frac{1}{k} \abovebelow{\vec{0}}{\vec{x}_1} \right) \in Q.
\eqno\qed
\]
\end{proof}

Note that \autoref{lem:fixed-point} does not transfer to
lasso programs: there might only be one fixed point and the stem
might exclude this point
(e.g., $a = -0.5$ and $b = 3.5$ in example \autoref{fig:IntroExA}).

Because fixed points give rise to trivial geometric nontermination arguments,
we can derive a criterion
for the existence of geometric nontermination arguments
from \autoref{lem:fixed-point}.

\begin{corollary}[Bounded Infinite Executions]
\label{cor:bounded-execution}
If the linear loop program $L = (true, \loopt)$ has a bounded infinite execution,
then it has a geometric nontermination argument of size $0$.
\end{corollary}
\begin{proof}
By \autoref{lem:fixed-point} there is a fixed point $\vec{x}^\ast$ such that
$(\vec{x}^\ast, \vec{x}^\ast) \in \loopt$.
We choose
$\vec{x_0} = \vec{x_1} = \vec{x}^\ast$
which satisfies (point) and (ray)
and thus is a geometric nontermination argument for $L$.
\qed
\end{proof}

\begin{example}\label{ex:strict}
Note that according to our definition of a linear lasso program,
the relation $\loopt$ is a topologically closed set.
If we allowed the formula defining $\loopt$ to also contain strict inequalities,
\autoref{lem:fixed-point} no longer holds:
the following program is nonterminating and has a bounded infinite execution,
but it does not have a fixed point.
However, the topological closure of the relation $\loopt$
contains the fixed point $a = 0$.
\begin{center}
\begin{minipage}{29mm}
\begin{lstlisting}
while (a > 0):
    a := a / 2;
\end{lstlisting}
\end{minipage}
\end{center}
Nevertheless, this example still has the geometric nontermination argument
$\vec{x_1} = 1$, $\vec{y_1} = -0.5$, $\lambda = 0.5$.
\end{example}

\subsection{Nonnegative Eigenvalues}
\label{ssec:nonnegative-eigenvalues}

This section is dedicated to the proof of the following completeness result
for deterministic linear loop programs.

\begin{theorem}[Completeness]\label{thm:completeness}
If a deterministic linear loop program $L$ of the form
\texttt{while ($G\vec x \leq \vec g$) do $\vec x := M\vec x + \vec m$}
with $n$ variables is nonterminating
and $M$ has only nonnegative real eigenvalues,
then there is a geometric nontermination argument for $L$ of size at most $n$.
\end{theorem}

To prove this completeness theorem,
we need to construct a GNTA from a given infinite execution.
The following lemma shows that we can restrict our construction
to exclude all linear subspaces that have a bounded execution.

\begin{lemma}[Loop Disassembly]
\label{lem:loop-disassembly}
Let $L = (true, \loopt)$ be a linear loop program over
$\mathbb{R}^n = \mathcal{U} \oplus \mathcal{V}$
where $\mathcal{U}$ and $\mathcal{V}$ are linear subspaces of $\mathbb{R}^n$.
Suppose $L$ is nonterminating and
there is an infinite execution
that is bounded when projected to the subspace $\mathcal{U}$.
Let $\vec{x}^\mathcal{U}$ be the fixed point in $\mathcal{U}$
that exists according to \autoref{lem:fixed-point}.
Then the linear loop program $L^{\mathcal{V}}$
that we get by projecting to the subspace $\mathcal{V} + \vec{x}^\mathcal{U}$
is nonterminating.
Moreover, if $L^{\mathcal{V}}$ has a GNTA of size $k$,
then $L$ has a GNTA of size $k$.
\end{lemma}
\begin{proof}
Without loss of generality,
we are in the basis of $\mathcal{U}$ and $\mathcal{V}$ so that
these spaces are nicely separated by the use of different variables.
Using the infinite execution of $L$ that is bounded on $\mathcal{U}$
we can do the construction from the proof of \autoref{lem:fixed-point}
to get an infinite execution $\vec{z_0}, \vec{z_1}, \ldots$
that yields the fixed point $\vec{x}^\mathcal{U}$
when projected to $\mathcal{U}$.
We fix $\vec{x}^\mathcal{U}$ in the loop transition by replacing all
variables from $\mathcal{U}$ with the values from $\vec{x}^\mathcal{U}$
and get the linear loop program $L^{\mathcal{V}}$
(this is the projection to $\mathcal{V} + \vec{x}^\mathcal{U}$).
Importantly, the projection of $\vec{z_0}, \vec{z_1}, \ldots$
to $\mathcal{V} + \vec{x}^{\mathcal{U}}$
is still an infinite execution,
hence the loop $L^{\mathcal{V}}$ is nonterminating.
Given a GNTA for $L^{\mathcal{V}}$
we can construct a GNTA for $L$
by adding the vector $\vec{x}^\mathcal{U}$ to $\vec{x_0}$ and $\vec{x_1}$.
\qed
\end{proof}

\begin{proof}[of \autoref{thm:completeness}]
The polyhedron corresponding to loop transition of
the deterministic linear loop program $L$ is
\begin{equation}\label{eq:det-loop}
\left(
\begin{matrix}
G & 0 \\
M & -I \\
-M & I
\end{matrix}
\right)
\left(
\begin{matrix}
\vec x \\
\vec{x'}
\end{matrix}
\right)
\leq
\left(
\begin{matrix}
\vec g \\
-\vec m \\
\vec m
\end{matrix}
\right).
\end{equation}
Define $\mathcal{Y}$ to be
the convex cone spanned by the rays of the guard polyhedron:
\[
\mathcal{Y} :=
\{ \vec y \in \mathbb{R}^n \mid G\vec y \leq 0 \}
\]
Let $\overline{\mathcal{Y}}$ be the smallest linear subspace of $\mathbb{R}^n$
that contains $\mathcal{Y}$,
i.e., $\overline{\mathcal{Y}} = \mathcal{Y} - \mathcal{Y}$
using pointwise subtraction,
and let $\overline{\mathcal{Y}}^\bot$ be the linear subspace of $\mathbb{R}^n$
orthogonal to $\overline{\mathcal{Y}}$;
hence $\mathbb{R}^n = \overline{\mathcal{Y}} \oplus \overline{\mathcal{Y}}^\bot$.

Let $P := \{ \vec x \in \mathbb{R}^n \mid G\vec x \leq \vec g \}$
denote the guard polyhedron.
Its projection $P^{\overline{\mathcal{Y}}^\bot}$
to the subspace $\overline{\mathcal{Y}}^\bot$
is again a polyhedron.
By the decomposition theorem for polyhedra~\cite[Cor.~7.1b]{Schrijver99},
$P^{\overline{\mathcal{Y}}^\bot} = Q + C$
for some polytope $Q$ and some convex cone $C$.
However, by definition of the subspace $\overline{\mathcal{Y}}^\bot$,
the convex cone $C$ must be equal to $\{ \vec 0 \}$:
for any $\vec y \in C \subseteq \overline{\mathcal{Y}}^\bot$,
we have $G\vec y \leq \vec 0$,
thus $\vec y \in \mathcal{Y}$,
and therefore $\vec y$ is orthogonal to itself,
i.e., $\vec y = \vec 0$.
We conclude that $P^{\overline{\mathcal{Y}}^\bot}$ must be a polytope,
and thus it is bounded.
By assumption $L$ is nonterminating,
so $L^{\overline{\mathcal{Y}}^\bot}$ is nonterminating,
and since $P^{\overline{\mathcal{Y}}^\bot}$ is bounded,
any infinite execution of $L^{\overline{\mathcal{Y}}^\bot}$
must be bounded.

Let $\mathcal{U}$ denote the direct sum of the generalized eigenspaces
for the eigenvalues $0 \leq \lambda < 1$.
Any infinite execution is necessarily bounded on the subspace $\mathcal{U}$
since on this space the map $\vec{x} \mapsto M\vec{x} + \vec{m}$ is a contraction.
Let $\mathcal{U}^\bot$ denote
the subspace of $\mathbb{R}^n$ orthogonal to $\mathcal{U}$.
The space $\overline{\mathcal{Y}} \cap \mathcal{U}^\bot$
is a linear subspace of $\mathbb{R}^n$ and
any infinite execution in its complement is bounded.
Hence we can turn our analysis to the subspace
$\overline{\mathcal{Y}} \cap \mathcal{U}^\bot + \vec{x}$
for some $\vec{x} \in \overline{\mathcal{Y}}^\bot \oplus \mathcal{U}$
for the rest of the proof
according to \autoref{lem:loop-disassembly}.
From now on,
we implicitly assume that we are in this space
without changing any of the notation.

\paragraph{Part 1.}
In this part we show that
there is a basis $\vec{y_1}, \ldots, \vec{y_k} \in \mathcal{Y}$
such that $M$ turns into a matrix $U$ of the form given in \eqref{eq:def-U}
with $\lambda_1, \ldots, \lambda_k, \mu_1, \ldots, \mu_{k-1} \geq 0$.
Since we allow $\mu_i$ to be positive between different eigenvalues
(\autoref{ex:U-is-not-Jordan} illustrates why),
this is not necessarily a Jordan normal form and
the vectors $\vec{y_i}$ are not necessarily generalized eigenvectors.

We choose a basis $\vec{v_1}, \ldots, \vec{v_k}$ such that
$M$ is in Jordan normal form with the eigenvalues ordered by size
such that the largest eigenvalues come first.
Define $\mathcal{V}_1 := \overline{\mathcal{Y}} \cap \mathcal{U}^\bot$
and let $\mathcal{V}_1 \supset \ldots \supset \mathcal{V}_k$ be
a strictly descending chain of linear subspaces where $\mathcal{V}_i$
is spanned by $\vec{v_i}, \ldots, \vec{v_k}$.

We define a basis $\vec{w_1}, \ldots, \vec{w_k}$
by doing the following for each Jordan block of $M$,
starting with $i = 1$.
Let $M^{(i)}$ be the projection of $M$ to the linear subspace $\mathcal{V}_i$
and let $\lambda$ be the largest eigenvalues of $M^{(i)}$.
The $m$-fold iteration of a Jordan block $J_\ell(\lambda)$
for $m \geq \ell$ is given by
\begin{equation}\label{eq:iterated-Jordan-block}
J_\ell(\lambda)^m =
\left(
\begin{matrix}
\lambda^m & \tbinom{m}{1}\lambda^{m-1}
                      & \dots & \tbinom{m}{\ell}\lambda^{m-\ell} \\
          & \lambda^m & \dots  & \tbinom{m}{\ell-1}\lambda^{m-\ell+1} \\
          &           & \ddots & \vdots \\
0         &           &        & \lambda^m
\end{matrix}
\right) \in \mathbb{R}^{\ell \times \ell}.
\end{equation}
Let $\vec{z_0}, \vec{z_1}, \vec{z_2}, \ldots$
be an infinite execution of the loop $L$
in the basis $\vec{v_i}, \ldots, \vec{v_k}$
projected to the space $\mathcal{V}_i$.
Since by \autoref{lem:loop-disassembly} we can assume that
there are no fixed points on this space,
$|\vec{z_t}| \to \infty$ as $t \to \infty$ in each of the top $\ell$ components.
Asymptotically, the largest eigenvalue $\lambda$ dominates and
in each row of $J_k(\lambda_i)^m$ \eqref{eq:iterated-Jordan-block},
the entries $\tbinom{m}{j}\lambda^{m-j}$ in the rightmost column
grow the fastest with an asymptotic rate of $\Theta(m^j \exp(m))$.
Therefore the sign of the component
corresponding to basis vector $\vec{v_{i+\ell}}$
determines whether the top $\ell$ entries tend to $+\infty$ or $-\infty$,
but the top $\ell$ entries of $\vec{z_t}$ corresponding to the top Jordan block
will all have the same sign eventually.
Because no state can violate the guard condition
we have that the guard cannot constraint the infinite execution in
the direction of $\vec{v_j}$ or $-\vec{v_j}$, i.e.,
$G^{\mathcal{V}_i} \vec{v_j} \leq \vec 0$
for each $i \leq j \leq i+\ell$ or
$G^{\mathcal{V}_i} \vec{v_j} \geq \vec 0$ for each $i \leq j \leq i+\ell$,
where $G^{\mathcal{V}_i}$ is the projection of $G$ to the subspace $\mathcal{V}_i$.
So without loss of generality the former holds
(otherwise we use $-\vec{v_j}$ instead of $\vec{v_j}$ for $i \leq j \leq i + \ell$) and
for $i \leq j \leq i+\ell$ we get
$\vec{v_j} \in \mathcal{Y} + \mathcal{V}_i^\bot$
where $\mathcal{V}_i^\bot$ is
the space spanned by $\vec{v_1}, \ldots, \vec{v_{i-1}}$.
Hence there is a $\vec{u_j} \in \mathcal{V}_i^\bot$
such that $\vec{w_j} := \vec{v_j} + \vec{u_j}$ is an element of $\mathcal{Y}$.
Now we move on to the subspace $\mathcal{V}_{i+\ell+1}$,
discarding the top Jordan block.

Let $T$ be the matrix $M$ written in
the basis $\vec{w_1}, \ldots, \vec{w_k}$.
Then $T$ is of upper triangular form:
whenever we apply $M\vec{w_i}$ we get $\lambda_i \vec{w_i} + \vec{u_i}$
($\vec{w_i}$ was an eigenvector in the space $\mathcal{V}_i$)
where $\vec{u_i} \in \mathcal{V}_i^\bot$,
the space spanned by $\vec{v_1}, \ldots, \vec{v_{i-1}}$
(which is identical with the space spanned by $\vec{w_1}, \ldots, \vec{w_{i-1}}$).
Moreover, since we processed every Jordan block entirely,
we have that
for $\vec{w_i}$ and $\vec{w_j}$ from the same generalized eigenspace
($T_{i,i} = T_{j,j}$) that for $i > j$
\begin{equation}\label{eq:T-on-the-same-geigenspace}
T_{j,i} \in \{ 0, 1 \}
\text{ and } T_{j,i} = 1 \text{ implies } i = j+1.
\end{equation}
In other words,
when projected to any generalized eigenspace
$T$ consists only of Jordan blocks.

Now we change basis again in order to get
the upper triangular matrix $U$ defined in \eqref{eq:def-U}
from $T$.
For this we define the vectors
\[
\vec{y_i} := \vec\beta_i \sum_{j=1}^i \alpha_{i,j} \vec{w_j}.
\]
with nonnegative real numbers $\alpha_{i,j} \geq 0$,
$\alpha_{i,i} > 0$, and $\vec\beta > 0$ to be determined later.
Define the matrices $W := (\vec{w_1} \ldots \vec{w_k})$,
$Y := (\vec{y_1} \ldots \vec{y_k})$, and
$\alpha := (\alpha_{i,j})_{1 \leq j \leq i \leq k}$.
So $\alpha$ is a nonnegative lower triangular matrix with a positive diagonal
and hence invertible.
Since $\alpha$ and $W$ are invertible,
the matrix $Y = \diag(\vec\beta) \alpha W$ is invertible as well and thus
the vectors $\vec{y_1}, \ldots, \vec{y_k}$ form a basis.
Moreover, we have
$\vec{y_i} \in \mathcal{Y}$ for each $i$
since $\alpha \geq 0$, $\vec\beta > 0$, and $\mathcal{Y}$ is a convex cone.
Therefore we get
\begin{equation}\label{eq:GY-nonpositive}
GY \leq 0.
\end{equation}

We will first choose $\alpha$.
Define $T =: D + N$
where $D = \diag(\lambda_1, \ldots, \lambda_k)$ is a diagonal matrix and
$N$ is nilpotent.
Since $\vec{w_1}$ is an eigenvector of $M$ we have
$M\vec{y_1}
= M \vec\beta_1 \alpha_{1,1} \vec{w_1}
= \lambda_1 \vec\beta_1 \alpha_{1,1} \vec{w_1}
= \lambda_1 \vec{y_1}$.
To get the form in \eqref{eq:def-U}, we need for all $i > 1$
\begin{equation}\label{eq:y-makes-U}
M\vec{y_i} = \lambda_i \vec{y_i} + \mu_{i-1} \vec{y_{i-1}}.
\end{equation}
Written in the basis $\vec{w_1}, \ldots, \vec{w_k}$
(i.e., multiplied with $W^{-1}$),
\[
(D + N) \vec\beta_i \sum_{j \leq i} \alpha_{i,j} \vec{e_j}
= \lambda_i \vec\beta_i \sum_{j \leq i} \alpha_{i,j} \vec{e_j}
  + \mu_{i-1} \vec\beta_{i-1} \sum_{j < i} \alpha_{i-1,j} \vec{e_j}.
\]
Hence we want to pick $\alpha$ such that
\begin{equation}\label{eq:constraints-on-alpha}
  \sum_{j \leq i} \alpha_{i,j} (\lambda_j - \lambda_i) \vec{e_j}
  + N \sum_{j \leq i} \alpha_{i,j} \vec{e_j}
  - \mu_{i-1} \vec\beta_{i-1} \sum_{j < i} \alpha_{i-1,j} \vec{e_j}
= \vec 0.
\end{equation}
First note that these constraints are independent of $\vec\beta$
if we set $\mu_{i-1} := \vec\beta_{i-1}^{-1} > 0$,
so we can leave assigning a value to $\vec\beta$ to a later part of the proof.

We distinguish two cases.
First, if $\lambda_{i-1} \neq \lambda_i$,
then $\lambda_j - \lambda_i$ is positive for all $j < i$
because larger eigenvalues come first.
Since $N$ is nilpotent and upper triangular,
$N \sum_{j \leq i} \alpha_{i,j} \vec{e_j}$ is a linear combination of
$\vec{e_1}, \ldots, \vec{e_{i-1}}$
(i.e., only the first $i-1$ entries are nonzero).
Whatever values this vector assumes,
we can increase the parameters $\alpha_{i,j}$ for $j < i$
to make \eqref{eq:constraints-on-alpha} larger
and increase the parameters $\alpha_{i-1,j}$ for $j < i$
to make \eqref{eq:constraints-on-alpha} smaller.

Second, let $\ell$ be minimal such that $\lambda_\ell = \lambda_i$
with $\ell \neq i$,
then $\vec{w_\ell}, \ldots, \vec{w_j}$ are from the same generalized eigenspace.
For the rows $1, \ldots, \ell-1$ we can proceed as we did in the first case and
for the rows $\ell, \ldots, i-1$ we note that
by \eqref{eq:T-on-the-same-geigenspace}
$N \vec{e_j} = T_{j-1,j} \vec{e_{j-1}}$.
Hence the remaining constraints \eqref{eq:constraints-on-alpha} are
\[
  \sum_{\ell < j \leq i} \alpha_{i,j} T_{j-1,j} \vec{e_{j-1}}
  - \mu_{i-1} \sum_{\ell \leq j < i} \alpha_{i-1,j} \vec{e_j}
= \vec 0,
\]
which is solved by $\alpha_{i,j+1} T_{j,j+1} = \alpha_{i-1,j}$ for $\ell \leq j < i$.
This is only a problem if there is a $j$ such that $T_{j-1,j} = 0$,
i.e., if there are multiple Jordan blocks for the same eigenvalue.
In this case, we can reduce the dimension of the generalized eigenspace to
the dimension of the largest Jordan block by combining all Jordan blocks:
if $M\vec{y_i} = \lambda \vec{y_i} + \vec{y_{i-1}}$, and
$M\vec{y_j} = \lambda \vec{y_j} + \vec{y_{j-1}}$, then
$M(\vec{y_i} + \vec{y_j})
= \lambda (\vec{y_i} + \vec{y_j})
+ (\vec{y_{i-1}} + \vec{y_{j-1}})$ and
if $M\vec{y_i} = \lambda \vec{y_i} + \vec{y_{i-1}}$, and
$M\vec{y_j} = \lambda \vec{y_j}$, then
$M(\vec{y_i} + \vec{y_j})
= \lambda (\vec{y_i} + \vec{y_j})
+ \vec{y_{i-1}}$.
In both cases we can replace the basis vector $\vec{y_i}$ with
$\vec{y_i} + \vec{y_j}$ without reducing the expressiveness of the GNTA.

Importantly, there are no cyclic dependencies in the values of $\alpha$
because neither one of the coefficients $\alpha$ can be made too large.
Therefore we can choose $\alpha \geq 0$ such that
\eqref{eq:y-makes-U} is satisfied
for all $i > 1$ and hence the basis $\vec{y_1}, \ldots, \vec{y_k}$
brings $M$ into the desired form \eqref{eq:def-U}.

\paragraph{Part 2.}
In this part we construct the geometric nontermination argument and
check the constraints from \autoref{def:gnta}.
Since $L$ has an infinite execution,
there is a point $\vec x$ that fulfills the guard, i.e., $G\vec x \leq \vec g$.
We choose $\vec{x_1} := \vec x + Y\vec\gamma$
with $\vec\gamma \geq \vec 0$ to be determined later.
Moreover, we choose $\lambda_1, \ldots, \lambda_k$ and
$\mu_1, \ldots, \mu_{k-1}$ from the entries of $U$ given in \eqref{eq:def-U}.
The size of our GNTA is $k$,
the number of vectors $\vec{y_1}, \ldots, \vec{y_k}$.
These vectors form a basis of $\overline{\mathcal{Y}} \cap \mathcal{U}^\bot$,
which is a subspace of $\mathbb{R}^n$;
thus $k \leq n$, as required.

The constraint
\hyperref[itm:gnta-domain]{(domain)} is satisfied by construction and
the constraint \hyperref[itm:gnta-init]{(init)} is vacuous
since $L$ is a loop program.
For \hyperref[itm:gnta-ray]{(ray)} note that
from \eqref{eq:GY-nonpositive} and \eqref{eq:y-makes-U} we get
\[
\left(
\begin{matrix}
G & 0 \\
M & -I \\
-M & I
\end{matrix}
\right)
\left(
\begin{matrix}
\vec{y_i} \\
\lambda_i \vec{y_i} + \mu_{i-1} \vec{y_{i-1}}
\end{matrix}
\right)
\leq
\left(
\begin{matrix}
\vec 0 \\
\vec 0 \\
\vec 0
\end{matrix}
\right).
\]
The remainder of this proof shows that
we can choose $\vec\beta$ and $\vec\gamma$ such that
\hyperref[itm:gnta-point]{(point)} is satisfied, i.e., that
\begin{equation}\label{eq:x1-and-beta}
     G\vec{x_1}
\leq \vec g
\text{ and }
  M \vec{x_1} + \vec{m}
= \vec{x_1} + Y\vec{1}.
\end{equation}

The vector $\vec{x_1}$ satisfies the guard since
$G\vec{x_1} = G\vec x + G Y \vec\gamma \leq \vec g + \vec 0$
according to \eqref{eq:GY-nonpositive},
which yields the first part of \eqref{eq:x1-and-beta}.
For the second part we observe the following.
\begin{align}
&~ &
     M\vec{x_1} + \vec m
  &= \vec{x_1} + Y\vec{1} \notag \\
&\Longleftrightarrow\quad &
     (M - I)(\vec x + Y\vec\gamma) + \vec m
  &= Y\vec{1} \notag \\
&\Longleftrightarrow\quad &
     (M - I)\vec x + \vec m
  &= Y\vec{1} - (M - I) Y\vec\gamma \notag \\
\intertext{Since $Y$ is a basis, it is invertible, so}
&\Longleftrightarrow\quad &
     Y^{-1}(M - I)\vec x + Y^{-1}\vec m
  &= \vec{1} - Y^{-1}(M - I) Y\vec\gamma
\notag \\
&\Longleftrightarrow\quad &
     (U - I) Y^{-1}\vec x + Y^{-1}\vec m
  &= \vec{1} - (U - I)\vec\gamma
\notag \\
&\Longleftrightarrow\quad &
     (U - I)\vec{\tilde x} + \vec{\tilde m}
  &= \vec{1} - (U - I)\vec\gamma
\label{eq:constraint-point}
\end{align}
with $\vec{\tilde x} := Y^{-1}\vec x = W^{-1} \alpha^{-1} \diag(\vec\beta)^{-1} \vec x$ and
$\vec{\tilde m} := Y^{-1} \vec m = W^{-1} \alpha^{-1} \diag(\vec\beta)^{-1} \vec m$.
Equation~\eqref{eq:constraint-point} is now conveniently
in the basis $\vec{y_1}, \ldots, \vec{y_k}$
and all that remains to show is that
we can choose $\vec\gamma \geq \vec 0$ and $\vec\beta > 0$
such that \eqref{eq:constraint-point} is satisfied.

We proceed for each (not quite Jordan) block of $U$ separately,
i.e., we assume that
we are looking at the subspace $\vec{y_j}, \ldots, \vec{y_i}$
with $\mu_i = \mu_{j-1} = 0$ and $\mu_k > 0$ for all $j \leq k < i$.
If this space only contains eigenvalues that are larger than $1$,
then $U - I$ is invertible
and has only nonnegative entries.
By using large enough values for $\vec\beta$,
we can make $\vec{\tilde x}$ and $\vec{\tilde m}$ small enough,
such that $\vec 1 \geq (U - I)\vec{\tilde x} + \vec{\tilde m}$.
Then we just need to pick $\vec\gamma$ appropriately.

If there is at least one eigenvalue $1$,
then $U - I$ is not invertible,
so \eqref{eq:constraint-point} could be overconstraint.
Notice that $\mu_\ell > 0$ for all $j \leq \ell < i$,
so only the bottom entry in the vector equation \eqref{eq:constraint-point}
is not covered by $\vec\gamma$.
Moreover, since eigenvalues are ordered in decreasing order and
all eigenvalues in our current subspace are $\geq 1$,
we conclude that the eigenvalue for the bottom entry is $1$.
(Furthermore, $i$ is the highest index
since each eigenvalue occurs only in one block).
Thus we get the equation $\vec{\tilde m}_i = 1$.
If $\vec{\tilde m}_i$ is positive, this equation has a solution
since we can adjust $\vec\beta_i$ accordingly.
If it is zero, then the execution on the space spanned by $\vec{y_i}$ is bounded,
which we can rule out by \autoref{lem:loop-disassembly}.

It remains to rule out that $\vec{\tilde m}_i$ is negative.
Let $\mathcal{U}$ be the generalized eigenspace to the eigenvector $1$
and use \autoref{lem:eigenvalue-1} below to conclude that
$\vec{o} := N^{k-1}\vec{m} + \vec{u} \in \mathcal{Y}$
for some $\vec{u} \in \mathcal{U}^\bot$.
We have that
$M\vec{o} = M(N^{k-1}\vec{m} + \vec{u}) = M\vec{u} \in \mathcal{U}^\bot$,
so $\vec{o}$ is a candidate to pick for the vector $\vec{w_i}$.
Therefore without loss of generality we did so in part 1 of this proof
and since $\vec{y_i}$ is in the convex cone
spanned by the basis $\vec{w_1}, \ldots, \vec{w_k}$
we get $\vec{\tilde m}_i > 0$.
\qed
\end{proof}

\begin{lemma}[Deterministic Loops with Eigenvalue 1]
\label{lem:eigenvalue-1}
Let $M = I + N$ and let $N$ be nilpotent with nilpotence index $k$
($k := \min \{ i \mid N^i = 0 \}$).
If $GN^{k-1} \vec m \not\leq \vec 0$, then $L$ is terminating.
\end{lemma}
\begin{proof}
We show termination
by providing an $k$-nested ranking function~\cite[Def.\ 4.7]{LH15LMCS}.
By \cite[Lem.\ 3.3]{LH15LMCS} and \cite[Thm.\ 4.10]{LH15LMCS},
this implies that $L$ is terminating.

According to the premise, $G N^{k-1} \vec m \not\leq 0$,
hence there is at least one positive entry in the vector $G N^{k-1} \vec m$.
Let $\vec h$ be a row vector of $G$ such that
$\vec h^T N^{k-1} \vec m =: \delta > 0$,
and let $h_0 \in \mathbb{R}$ be the corresponding entry in $\vec g$.
Let $\vec x$ be any state and
let $\vec{x'}$ be a next state after the loop transition,
i.e., $\vec{x'} = M \vec x + \vec m$.
Define the affine-linear functions
$f_j(\vec x) := -\vec h^T N^{k-j} \vec x + c_j$ for
$1 \leq j \leq k$ with constants $c_j \in \mathbb{R}$ to be determined later.
Since every state $\vec x$ satisfies the guard we have $\vec h^T\vec x \leq h_0$,
hence
$f_k(\vec x) = -\vec h^T \vec x + c_k \geq - h_0 + c_k > 0$ for $c_k := h_0 + 1$.
\begin{align*}
   f_1(\vec{x'})
 = f_1(\vec x + N\vec x + \vec m)
&= -\vec h^T N^{k-1} (\vec x + N \vec x + \vec m) + c_1 \\
&= f_1(\vec x) - \vec h^T N^k \vec x - \vec h^T N^{k-1} \vec m \\
&< f_1(\vec x) - 0 - \delta
\end{align*}
For $1 < j \leq k$,
\begin{align*}
   f_j(\vec{x'})
 = f_j (\vec x + N \vec x + \vec m)
&= -\vec h^T N^{k-j} (\vec x + N \vec x + \vec m) + c_j \\
&= f_j(\vec x) + f_{j-1}(\vec x) - \vec h^T N^{k-j} \vec m - c_{j-1} \\
&< f_j(\vec x) + f_{j-1}(\vec x)
\end{align*}
for $c_{j-1} := -\vec h^T N^{k-j} \vec m - 1$.
\qed
\end{proof}

\begin{example}[$U$ is not in Jordan Form]
\label{ex:U-is-not-Jordan}
The matrix $U$ defined in \eqref{eq:def-U}
and used in the completeness proof is generally \emph{not}
the Jordan normal form of the loop's transition matrix $M$.
Consider the following linear loop program.
\begin{center}
\begin{minipage}{45mm}
\begin{lstlisting}
while ($a - b \geq 0 \land b \geq 0$):
	$a$ := $3a$;
	$b$ := $b + 1$;
\end{lstlisting}
\end{minipage}
\end{center}
This program is nonterminating because $a$ grows exponentially and
hence faster than $b$.
It has the geometric nontermination argument
\begin{align*}
\vec{x_0} &= \abovebelow{9}{1}, &
\vec{x_1} &= \abovebelow{9}{1}, &
\vec{y_1} &= \abovebelow{12}{0}, &
\vec{y_2} &= \abovebelow{6}{1}, &
\lambda_1 &= 3, &
\lambda_2 &= 1, &
\mu_1 &= 1.
\end{align*}
The matrix corresponding to the linear loop update is
\[
M =
\left(
\begin{matrix}
3 & 0 \\
0 & 1
\end{matrix}
\right)
\]
which is diagonal (hence diagonalizable).
Therefore $M$ is already in Jordan normal form.
The matrix $U$ defined according to \eqref{eq:def-U} is
\[
U =
\left(
\begin{matrix}
3 & 1 \\
0 & 1
\end{matrix}
\right).
\]
The nilpotent component $\mu_1 = 1$ is important and
there is no GTNA for this loop program where $\mu_1 = 0$ since
the eigenspace to the eigenvector $1$ is spanned by $(0\; 1)^T$
which is in $\overline{\mathcal{Y}}$, but not in $\mathcal{Y}$.
\end{example}


\section{Experiments}
\label{sec:experiments}

We implemented our method in a tool that is specialized for the
analysis of lasso programs
and called \textsc{Ultimate LassoRanker}.
\textsc{LassoRanker} is used by \textsc{Ultimate Büchi Automizer}
which analyzes termination of (general) C programs via the following approach~\cite{HHP14}.
\textsc{Büchi Automizer} iteratively picks lasso shaped paths in the control
flow graph converts them to lasso programs and lets \textsc{LassoRanker} analyze them.
In case \textsc{LassoRanker} was able to prove nontermination a real counterexample to termination was found,
in case \textsc{LassoRanker} was able to provide a termination argument (e.g., a linear ranking function),
Büchi Automizer continues the analysis, but only on lasso shaped paths
for which the termination arguments obtained in former iterations are not applicable.

We applied \textsc{Ultimate Büchi Automizer} to the 631 benchmarks
from the category Termination of the
5th software verification competition SV-COMP 2016~\cite{svcomp16}
in two different settings:
one setting where we use our geometric nontermination arguments (GNTA) and
one setting where we only synthesize fixed points (i.e., infinite executions where one state is always repeated).

In both settings the constraints were stated over the integers and
we used the SMT solver Z3~\cite{JovanovicM12} with a timeout of 12s to solve our constraints.
The overall timeout for the termination analysis was 60s.
Using the fixed point setting the tool was able to solve 441 benchmarks
and the overall time for solving the (linear) constraints was 56 seconds.
Using the GNTA setting the tool was able to solve 487 benchmarks
and the overall time for solving the (nonlinear) constraints 2349 seconds.
The GNTA setting was able to solve 47 benchmarks
that could not be solved using the fixed point setting because there was
a nonterminating execution which did not had a fixed point.
The fixed point setting was able to solve 2 benchmarks that could not be solved using the GNTA setting because in the latter setting solving the linear constraints took too long.


\section{Related Work}
\label{sec:related-work}

One line of related work is focused on decidability questions for deterministic
lasso programs.
Tiwari~\cite{Tiwari04} considered linear loop programs over the reals where only strict
inequalities are used in the guard and proved that termination is decidable.
Braverman~\cite{Braverman06} generalized this result to loop programs that use
strict and non-strict inequalities in the guard.
Furthermore, he proved that termination is also decidable for homogeneous
deterministic loop programs over the integers.
Rebiha et al.~\cite{RMM14} generalized the result to integer loops where
the update matrix has only real eigenvalues.
Ouaknine et al.~\cite{OPW15} generalized the result to integer lassos where
the update matrix of the loop is diagonalizable.

Another line of related work is also applicable to nondeterministic programs and uses a
constraint-based synthesis of recurrence sets.
The recurrence sets are defined by templates~\cite{VelroyenR08,GHMRX08}
or the constraint is given in a second order theory for bit vectors~\cite{DavidKL15}.
These approaches can be used to find nonterminating lassos that do not have
a geometric nontermination argument; however, this comes at the price that
for nondeterministic programs an $\exists\forall\exists$-constraint
has to be solved.

Furthermore, there is a long line of
research~\cite{foveoos/BrockschmidtSOG11,cav/AtigBEL12,CCFNO14,fmcad/CookFNO14,LarrazNORR14,DavidKL15,LeQC15,seahorn16}
that addresses programs that
are more general than lasso programs.


\section{Conclusion}
\label{sec:conclusion}

We presented a new approach to nontermination analysis for linear lasso programs.
This approach is based on geometric nontermination arguments,
which are an explicit representation of an infinite execution.
These nontermination arguments can be found by solving a set of
nonlinear constraints.
In \autoref{sec:completeness} we showed that the class of nonterminating
linear lasso programs that have a geometric nontermination argument
is quite large:
it contains at least
every deterministic linear loop program whose eigenvalues are nonnegative.
We expect that this statement can be extended
to encompass also negative and complex eigenvalues.

The synthesis of nontermination arguments is useful not only
to discover infinite loops in program code,
but also to accelerate termination analysis
(a nonterminating lasso does not need to be checked for a termination argument)
or overflow analysis.
Furthermore,
since geometric nontermination arguments readily provide an infinite execution,
a discovered software fault is transparent to the user.


\bibliographystyle{abbrv}
\bibliography{references}


\appendix
\clearpage
{\Huge\bf Appendix}

\section{List of Notation}
\label{app:notation}

\begin{tabular}{ll}
$\mathbb{R}$
	& the set of real numbers \\
$L$
	& a linear lasso program \\
$\vec0$
	& a vector of zeros (of the appropriate dimension) \\
$\vec1$
	& a vector of ones (of the appropriate dimension) \\
$\vec{e_i}$
	& the $i$-th unit vector \\
$n$ & number of variables in the loop program \\
$k,i,j,\ell$
	& natural numbers \\
$\alpha, \beta, \gamma$
	& various nonnegative/positive parameters \\
$G$
	& part of the program guard $G\vec x \leq \vec g$ \\
$M$
	& the linear map in a deterministic linear loop program \\
$\vec{x}$
	& a program state, i.e., a real-valued vector of dimension $n$ \\
$\vec{x^*}$
	& a fixed point of the loop transition \\
$N$
	& a nilpotent matrix \\
$D$
	& a diagonal matrix \\
$T, U$
	& upper triangular matrixes \\
$\lambda_i$
	& an eigenvalue \\
$\vec{y}$
	& a ray of the guard polyhedron, $G\vec{y} \leq 0$
\end{tabular}

\end{document}